%% file: main.tex
\documentclass[a4paper,UKenglish]{lipics-v2018}
\usepackage{bbm}
\nolinenumbers 

\subjclass{
}


\acknowledgements{The author thanks Benny Kimelfeld for suggesting this research topic, and for many helpful discussions.}

\title{The Complexity of the Possible Winner Problem over Partitioned Preferences}

\author{Batya Kenig}{Technion}{batyak@cs.technion.ac.il}{}{}

\titlerunning{The complexity of the Possible Winner problem over partitioned preferences}

\authorrunning{Batya Kenig}

\Copyright{Batya Kenig}

\usepackage{times}
\usepackage{xcolor}
\usepackage{soul}
\usepackage[utf8]{inputenc}
\usepackage{amsmath}
\usepackage{amsthm}
\usepackage{graphicx}
\usepackage{nicefrac}
\input{macros.tex}

\begin{document}

\maketitle

\begin{abstract} 
The $\PW$ problem asks, given an election where the voters' preferences over the set of candidates is partially specified, whether a distinguished candidate can become a winner.
In this work, we consider the computational complexity of $\PW$ under the assumption that the voter preferences are \e{partitioned}. That is, we assume that every voter provides a complete order over sets of incomparable candidates (e.g., candidates are ranked by their level of education). We consider elections with partitioned profiles over positional scoring rules, with an unbounded number of candidates, and unweighted voters.
Our first result is a polynomial time algorithm for voting rules with $2$ distinct values, which include the well-known $k$-approval voting rule.
We then go on to prove NP-hardness for a class of rules that contain all voting rules that produce scoring vectors with at least $4$ distinct values.
\eat{
For the remaining set of positional scoring rules, we prove that the $\PW$ problem remains NP-complete for all scoring rules  except 
 vectors of the form $(\underbrace{2,\dots,2}_{k_2},1,\dots,1,\underbrace{0,\dots,0}_{k_0})$ where $k_0$ and $k_2$ are fixed constants such that $k_0+k_2 >2$, for which the complexity remains open. 
}
\end{abstract}

\section{Introduction}\label{sec:intro}
In political elections, web site rankings, and multiagent systems, preferences of different parties (\e{voters}) have to be aggregated to form a joint decision. A general solution to this problem is to have the agents vote over the alternatives. The voting process is conducted as follows: each agent provides a ranking of the possible alternatives (\e{candidates}). Then, a \e{voting rule} takes these rankings as input and produces a set of chosen alternatives (\e{winners}) as output.
However, in many real-life settings one has to deal with \e{partial votes}: Some voters may have preferences over only a subset of the candidates. The $\PW$ problem, introduced by Konczak and Lang~\cite{Konczak05votingprocedures} \eat{\citeauthor{Konczak05votingprocedures}~[\citeyear{Konczak05votingprocedures}]} is defined as follows: Given a partial order for each of the voters, can a distinguished candidate $c$ win for at least one extension of the partial orders to linear ones ?

The answer to the $\PW$ problem depends on the voting rule that is used. In this work we consider \e{positional scoring rules}. A positional scoring rule provides a score value for every position that a candidate may take within a linear order, given as a scoring vector of length $m$ in the case of $m$ candidates. The scores of the candidates are added over all votes and the candidates with the maximal score win. For example, the $k$-approval voting rule, typically used in political elections, defined by $(1,\dots,1,0,\dots,0)$ starting with $k$ ones, enables voters to express their preference for $k$ candidates. Two popular special cases of $k$-approval are \e{plurality}, defined by $(1,0,\dots,0)$, and \e{veto}, defined by $(1,\dots,1,0)$.

The $\PW$ problem has been investigated for many types of voting systems~\cite{Betzler:2009:MCA:1661445.1661455,Lang07winnerdetermination,PINI20111272,DBLP:journals/jair/XiaC11}. For positional scoring rules,
Betzler and Dorn~\cite{DBLP:journals/jcss/BetzlerD10}
proved a result that was just one step away from a full dichotomy for the $\PW$ problem with positional scoring rules, unweighted votes, and any number of candidates. In particular, they showed NP-completeness for all but three scoring rules, namely plurality, veto, and the rule with the scoring vector $(2,1,\dots,1,0)$.  For plurality and veto, they showed that the problem is solvable in polynomial time, but the complexity of $\PW$ remained open for the scoring rule $(2,1,\dots,1,0)$ until it was shown to be NP-complete as well by
Baumeister and Rothe~\cite{DBLP:journals/corr/abs-1108-4436}.
\eat{~\citeauthor{DBLP:journals/corr/abs-1108-4436}~(\citeyear{DBLP:journals/corr/abs-1108-4436}).}

\e{Partitioned preferences} provide a good compromise between complete orders and arbitrary partial orders. 
Intuitively, the user provides a complete order over \e{sets} of incomparable items.
In the machine learning community, partitioned preferences were shown to be common in many real-life datasets, and have been used for learning statistical models on full and partial rankings~\cite{lebanon2008npm,DBLP:journals/jmlr/LuB14,DBLP:journals/jair/HuangKG12}. \eat{In particular, partitioned preferences generalize the known top-$t$, and bottom-$t$ preferences where a user specifies her most (least) preferred $t$ items.}

\eat{
Formally, in a partitioned preference, the set of items $A$ is partitioned into disjoint subsets $A_1,\dots,A_q$ such that (1) for all $i<j\leq q$, if $x \in A_i$ and $y\in A_j$ then $x \succ y$; and  (2) for each $i\leq q$ the alternatives in $A_i$ are incomparable.
}

In many scenarios, the user preferences are inherently partitioned. 
In recommender systems, the items are often partitioned according to their numerical level of desirability~\cite{Sarwar:2001:ICF:371920.372071} (e.g., the common star-rating system, where the scores range between $1$ and $5$ stars). In such a scenario, all items with identical scores are incomparable.
In some e-commerce systems, user preferences are obtained by tracking the various actions users perform~\cite{Koren:2011:OOM:2043932.2043956}. For example, searching or browsing a product is indicative of weak interest. Bookmarking it is indicative of stronger interest, followed by entering the product to the ``shopping cart''. Finally, the strongest indication would be actually purchasing the product. In this case as well, the items are partitioned into groups, where the desirability of each group is determined by its set of associated actions, and items in a common group are considered incomparable.
In the field of information retrieval, \e{learning to rank}~\cite{DBLP:conf/icml/CaoQLTL07,Liu:2009:LRI:1618303.1618304} refers to the process of applying machine learning techniques to rank a set of documents according to their relevance to a given query. In this setting, document scores are indicative of relevance to the query, and documents with identical scores are considered incomparable.

In this work we investigate the computational complexity of the $\PW$ problem with partitioned preference profiles.
Our first result is that determining the possible winner can be performed in polynomial time for $2$-valued voting rules (i.e., that produce scoring vectors with $2$ distinct values), which include the $k$-approval voting rule. We then show that our algorithm also solves the possible winner problem for the $(2,1,\dots,1,0)$ voting rule. These result are surprising because both of these rules are NP-complete when the partitioned assumption is dropped~\cite{DBLP:journals/jcss/BetzlerD10,DBLP:journals/corr/abs-1108-4436}. We then go on and prove hardness for the class of voting rules that produce scoring vectors containing at least $4$ distinct values, and a large class of voting rules with $3$ distinct values. 
The hardness proofs are involved because many of the order restrictions applied in the reductions for the general case are unavailable under the constraint of partitioned preferences.

\eat{
the remaining scoring rules except those that are restricted to three values (i.e., $2$, $1$, and $0$) for which the range of both score values $2$ and $0$ are bounded.
Formally, we prove NP-completeness for all scoring rules, except those that correspond to scoring vectors of the form:
\[
(\underbrace{2,\dots,2}_{k_2},1,\dots,1,\underbrace{0,\dots,0}_{k_0})
\]
where $k_2$ and $k_0$ are both constants such that $k_2+k_0>2$.
}
\section{Preliminaries} \label{section:preliminaries}
In this section
we present some basic notation and terminology that we use throughout
the manuscript.

\subsection{Orders and rankings}
A \e{partially ordered set} is a binary relation $\succ$ over a set of \e{alternatives}, or \e{candidates} $\candidates$ that satisfies transitivity ($a\succ b$ and
$b\succ c$ implies $a\succ c$) and irreflexivity ($a\succ a$ never
holds).  A \e{linear} (or \e{total}) order is a partially ordered set where every two items are comparable. 
We say that a total order $\succ_t$ \e{extends} the partial order $\succ_o$ if, for every pair of alternatives, $a,b$ such that $a\succ_o b$ it also holds that $a\succ_t b$.
We denote by
$\lin(\candidates)$ the set of all linear orders over $\candidates$, and by
$\lin(\candidates\midd{\succ_o})$ the set of linear orders over $\candidates$ that extend $\succ_o$.
In this manuscript we consider a special type of partial order termed \e{partitioned preferences}.
\begin{definition}[Partitioned preferences~\cite{LebanonM08}]\label{def:partitionedPrefs}
	A partial order $\succ_o$ is a \e{partitioned preference} if the set of candidates $\candidates$ can be partitioned into disjoint subsets $A_1,\dots,A_q$ such that: (1) for all $i<j\leq q$, if $c\in A_i$ and $c' \in A_j$ then $c \succ_o c'$; and (2) for each $i \leq q$, candidates in $A_i$ are incomparable under $\succ_o$ (i.e., $a \not\succ_o b$ and $b \not\succ_o a$ for every $a,b \in A_i$).
\end{definition}

\subsection{Elections}
Let $\V=\set{v_1,\dots,v_n}$ be a set of voters, and $\candidates=\set{c_1,\dots,c_m}$ a set of candidates. Every voter $v_i$ has a preference, also denoted $v_i$, which is a linear order or \e{complete vote} over $\candidates$ (i.e., $v_i \in \lin(\candidates)$).  
A tuple of $n$ complete votes $\V=(v_1,\dots,v_n)\in  \lin(\candidates)^n$ is an $n$-voter preference profile.
The set of all preference profiles on $\candidates$ is denoted by $\P(\candidates)$.
A \emph{voting rule} is a function from the set of all profiles on $\candidates$ to the set of nonempty subsets of $\candidates$. Formally $r \colon \P(\candidates) \mapsto 2^{\candidates}\setminus \set{\emptyset}$. For a voting rule $r$, and a preference profile $\V=(v_1,\dots,v_n)$, we say that candidate $c\in \candidates$ \e{wins the election} (or just wins) if $r(\V)=\set{c}$, and \e{co-wins} if $c \in r(\V)$. We denote an election by the triple $\I=(\candidates,\V,r)$.

We now generalize the election to the case where some or all of the votes are partial orders over the candidates. We consider the election $\I'=(\candidates, \O,r)$ where the voter profile $\O=(o_1,\dots,o_n)$ is comprised of partial orders over the candidates. We say that a profile $\V=(v_1,\dots,v_n)$ \e{extends} the profile $\O$ if they have the same cardinality (i.e., $|\O|=|\V|$), and every vote $v_i$ is a linear order that extends the partial order $o_i$ (i.e., $v_i \in \lin(o_i)$).
We say that a partial preference profile $\O=(o_1,\dots,o_n)$ is \e{partitioned} if every one of its preferences is partitioned.
\begin{definition}[$r$-possible winner (co-winner)]
Given an election $\I=(\candidates,\O,r)$ where 
$\O$ is a profile of partial orders over the candidate set $\candidates$, and a distinguished candidate $c\in \candidates$, does there exist an extension $\V$ of $\O$ such that $c=r(\V)$ ($c \in r(\V)$)?
\end{definition}

\subsection{Positional scoring rules}
Let $\I=(\candidates,\V,r)$ denote an election with $m$ candidates and $n$ voters.
A positional scoring rule $r$ is defined by a sequence $(\vec{\balpha}_m)_{m\in \mathbb{N}^+}$ of $m$-dimensional scoring vectors $\vec{\balpha}_m=(\alpha_m,\alpha_{m-1}\dots,\alpha_1)$ where $\alpha_m \geq \alpha_{m-1}\geq \dots \geq \alpha_1$ are positive integers denoted \e{score values}, and $\alpha_m > \alpha_1$ for every $m \in \mathbb{N}^+$. 
A voting rule $r=(\vec{\balpha}_m)_{m\in \mathbb{N}^+}$ is \e{normalized} if for every $m \in \mathbb{N}^+$ there is no integer greater than one that divides all score values in $\vec{\balpha}_m$, and $\alpha_1=0$. Since these assumptions have been shown to be non-restrictive~\cite{HEMASPAANDRA200773,DBLP:journals/jcss/BetzlerD10} we will consider only normalized scoring vectors in this work.
We say that a positional scoring rule $r=(\vec{\balpha}_m)_{m\in \mathbb{N}^+}$ is \e{pure}~\cite{DBLP:journals/jcss/BetzlerD10,dey_et_al:LIPIcs:2017:8135} if for every $m\geq 2$, the scoring vector for $m$ candidates can be obtained from the scoring vector for $m-1$ candidates by inserting an additional score value at an arbitrary position such that the resulting vector meets the monotonicity constraint. 
\eat{We assume that all positional scoring rules $r=(\vec{\balpha}_m)_{m\in \mathbb{N}^+}$ considered in this paper are pure.}We note that for voting rules that are defined for a constant number of candidates, the possible winner problem can be decided in polynomial time~\cite{Conitzer:2007:EFC:1236457.1236461,Walsh:2007:UPE:1619645.1619648}.

Given a complete vote $v\in \V$, and a candidate $c \in \candidates$, we define the \e{score} of $c$ in $v$ by $s(v,c) := \alpha_j$ where $j$ is the position of $c$ in $v$. The score of candidate $c \in \candidates$ in a profile $\V=\set{v_1,\dots,v_n}$ is defined as $s(\V,c)=\sum_{i=1}^ns(v_i,c)$. Whenever the profile $\V$ is clear from the context, we write $s(c)$. A positional scoring rule selects as winners all candidates $c$ with the maximum score $s(c)$.

Some popular examples of positional scoring rules are \emph{Borda}, for which the scoring vector is $\left(m-1,m-2,\dots,0\right)$, \emph{plurality}, for which the scoring vector is $\left(1,0,0,\dots,0\right)$, \emph{veto}, for which the scoring vector is $\left(1,1,1,\dots,1,0\right)$, and $k$-approval $(1\leq k\leq m-1)$, for which the scoring vector is  $\left(\underbrace{1,\dots,1}_k,0,\dots,0\right)$.
We assume that the scoring vector, and thus the scores of the candidates, can be computed in polynomial time given a complete profile.

\section{Summary of Results}
In this manuscript we consider the possible winner problem over partitioned preferences (Definition~\ref{def:partitionedPrefs}). We assume that all positional scoring rules are normalized.

\begin{definition}[$K$-valued voting rule]\label{def:K-valued}
	We say that a positional scoring rule $r=(\vec{\balpha}_m)_{m \in \mathbb{N}^+}$ is \emph{$K$-valued} if there exists a number $n_0 \in \mathbb{N}^+$ such that for all $m \geq n_0$, the score vector $\vec{\balpha}_m$ contains exactly $K$ distinct values.	
\end{definition}
By this definition, the $k$-approval, veto, and plurality voting rules are $2$-valued, while Borda has an unbounded number of different score values.

\begin{definition}[unbounded-value voting rule]\label{def:unbounded-value}
	We say that a positional scoring rule $r=(\vec{\balpha}_m)_{m \in \mathbb{N}^+}$ has an \e{unbounded number of positions with equal score values} if, for every $l \in \mathbb{N}^+$, there exists a number $n_0 \in \mathbb{N}^+$ such that for all $m \geq n_0$, the score vector $\vec{\balpha}_m$ contains at least $l$ consecutive positions $i+l-1,\dots,i$ where $\alpha_{i+l-1}=\dots=\alpha_i$.	
\end{definition}

\def\1thm{
		Let $r=(\vec{\balpha}_m)_{m \in \mathbb{N}^+}$ be a 
	positional scoring rule. Then we have the following when the preference profile is partitioned.
	\begin{enumerate}
		\item If $r$ is $2$-valued or if $r$ is $(2,1,\dots,1,0)$, then the $\PW$ problem over $r$ can be answered in polynomial time. 	
		\item If $r$ produces a scoring vector with at least $4$ distinct values then the $\PW$ problem is NP-complete for $r$.	
		\item If $r$ is $3$-valued, and $r$ produces a size-$m$ scoring vector that is differentiating, or where the number of positions occupied by either $\alpha_m$ or $\alpha_1$ is unbounded, then the $\PW$ problem is NP-complete for $r$.			
	\end{enumerate}
}

\def\2thm{
	Let $r=(\vec{\balpha}_m)_{m \in \mathbb{N}^+}$ be a pure, 
	$3$-valued positional scoring rule. Then we have the following.
	\begin{enumerate}
		\item If $r$ is equivalent to $(2,1,\dots,1,0)$ then the PW problem over $r$ can be answered in polynomial time. 		
		\item If, for every $m\in \mathbb{N}^+$, we have that the number of occurrences in $\vec{\balpha}_m$, for both values, $0$ and $2$, is at least $3$, then the PW problem for $r$ is NP-complete.
	\end{enumerate}
}

\eat{
\begin{definition}
	We say that a voting rule $r=(\vec{\balpha}_m)_{m \in \mathbb{N}^+}$ is $(2,1,\dots,1,0)$ if for every $m \geq 3$ we have that 	$\vec{\balpha}_m=(2,1,\dots,1,0)$.
\end{definition}
}
Let $\I=(\candidates, \O, r)$ denote an election where $\candidates$ is a set of candidates, $r=(\vec{\balpha}_m)_{m \in \mathbb{N}^+}$ is a positional scoring rule, and $\O$ is a partial profile where all of the votes are partitioned. In the rest of the manuscript we show the following.
If $r$ is $2$-valued, or if $r$ is $(2,1,\dots,1,0)$ then we show that the $\PW$ problem over $\I$ can be solved in polynomial time. In particular, this means that the $\PW$ problem is tractable for the $k$-approval voting rule. This result is surprising because it has been shown that when the partitioned assumption is dropped, the problem is intractable for both $k$-approval~\cite{DBLP:journals/jcss/BetzlerD10}, and $(2,1,\dots,1,0)$~\cite{DBLP:journals/corr/abs-1108-4436}.

Our hardness results, proved in Section~\ref{sec:hardness}, cover all scoring rules that produce scoring vectors with at least $4$ distinct values. For $3$-valued scoring rules, we prove hardness for all rules except vectors of the form $(\underbrace{2,\dots,2}_{k_2},1,\dots,1,\underbrace{0,\dots,0}_{k_0})$ where $k_0$ and $k_2$ are fixed constants such that $k_0+k_2 >2$, for which the complexity remains open. 
The main results are summarized in Theorem~\ref{theorem:main1}. A scoring rule is called \e{differentiating}~\cite{dey_et_al:LIPIcs:2017:8135} if it produces a scoring vector $\vec{\balpha}_m$  that contains two positions $i,j \geq 2$ where $j>i+1$ such that $(\alpha_j-\alpha_{j-1})>(\alpha_i-\alpha_{i-1})$.

\begin{theorem}\label{theorem:main1}
	\1thm
\end{theorem}
\eat{
Theorem~\ref{theorem:main1} covers all cases except $3$-valued scoring rules that produce scoring vectors $\vec{\balpha}_m$ in which both the portion of the vector occupied by $0$, and the portion occupied by $2$, are bounded by fixed constants $B_0$, and $B_2$ respectively, where $B_0+B_2 >2$.
\eat{
We note that in the general case, when the preference profile is not partitioned, the PW problem is NP-complete for both $k$-approval~\cite{DBLP:journals/jcss/BetzlerD10}, and $(2,1,\dots,1,0)$~\cite{DBLP:journals/corr/abs-1108-4436}.
}
}
\section{Tractability}\label{ref:Tractability}
In this section we describe a network flow algorithm that solves the possible winner problem in polynomial time for the $k$-approval and $(2,1,\dots,1,0)$ rules, when the preference profile is partitioned. 
Since we assume that the scoring vectors are normalized, then this algorithm is applicable to all $2$-valued scoring rules.
Some of the proofs in this section are deferred to the appendix.
\paragraph*{Maximal Scores}
Given a partial order $o \in \O$, and a candidate $c \in \candidates$, we denote by $\smax(o,c)$ the maximum score that candidate $c$ can obtain in any linear extension $v$ of $o$. That is, $\smax(o,c):=\max_{v \in \lin(o)}s(v,c)$.
It is straightforward to see that the maximum score of $c$ in any extension of $o$ is determined by the cardinality of the set of candidates that are preferred to it in $o$. That is, \[\smax(o,c)=\vec{\balpha}\left(\left|c' \in \candidates \mid c' \succ_o c \right|+1\right)\]
where $\vec{\balpha}$ is the scoring vector.
We denote by $\smax(\O,c)$ the maximum score that candidate $c$ can obtain in any extension of the partial profile $\O$ to a complete profile. It is straightforward to see that this score can be obtained by maximizing the score for each partial vote independently. Therefore:
\[
\smax(\O,c)=\sum_{i=1}^n\smax(o_i,c)
\]
When the partial profile $\O$ is clear from the context then we refer to this score as $\smax(c)$. 

In many cases it is convenient to \e{fix} the position of the distinguished candidate $c\in \candidates$ in the partial votes $\O$ such that its score is maximized. Formally, let $\O=(o_1,\dots,o_n)$ denote the partial vote. We denote by $\O^c=(o_1^c,\dots,o_n^c)$ the partial profile that is consistent with $\O$, and where the position of $c$ is fixed at the topmost position in each vote. Then:
\[
o_i^c=o_i\cup\set{c \succ c' \mid c'\not\succ_{o_i}c}
\]
In this case, the score of $c$ in any extension $\V$ of $\O^c$ is $s(\V,c)=s(\O^c,c)=\smax(\O,c)$.

\paragraph*{Elections with Partitioned Preferences}
Let $\O=\set{o_1,\dots,o_n}$ be a partitioned partial profile on $\candidates$.
Recall that $\O$ is a partitioned profile if all preferences in $\O$ are partitioned.
Lemma~\ref{lem:fixcscore} below shows that for deciding whether a distinguished candidate $c$ is a possible winner over a profile of partitioned preferences, we may restrict our attention to extensions of the profile $\O^c$ where the position (and score) of $c$ is fixed to the top of its partition. This is not the case when the profile is not limited to partitioned preferences as shown in the following example. The proof of Lemma~\ref{lem:fixcscore} is deferred to the appendix.

\begin{example}
	We consider the election $\I=(\candidates,\O,r)$ where $\candidates=\set{a,b,c,d}$, $\O=\set{o_1,\dots,o_5}$, and $r$ is the positional scoring rule corresponding to the vector $\
	\vec{\balpha}_4=(3,1,1,0)$. We consider the problem of deciding whether candidate $a$ is a possible winner.
	The votes are as follows.
	\begin{center}
		\begin{tabular}{ c|c } 
			$o_1$ & $a\succ b\succ c \succ d$ \\ 
			$o_2$ & $a\succ b\succ c \succ d$ \\ 
			$o_3$ & $b\succ a\succ c \succ d$ \\ 
			$o_4$ & $b\succ a\succ c \succ d$ \\ 
			$o_5$ & $b\succ a$ \\
		\end{tabular}
	\end{center}
	Let $\V_1$ denote an extension of $\O$ in which $v_5=(c\succ b \succ a \succ e)$. For $\V_1$ we have that $s(a,\V_1)=s(b,\V_1)=9$ making $a$ a possible co-winner. Now consider the  extension $\V_2$ in which $v_5=(b \succ a \succ d \succ c)$. For $\V_2$ we have that $s(a,\V_2)=9$, and $s(b,\V_2)=11$. Likewise, in the extension $\V_3$ in which $v_5=(b \succ a \succ c \succ d)$, $b$ is, again, the winner of the election. So we see that despite the fact that $a$ is a possible co-winner in $\I$, it is not the possible co-winner if positioned at its highest ranking position in every vote.
	\eat{So despite the fact that the position of $a$ has increased in $\V_2$, it has lost its winner status.}
\end{example}

\def\lemmafixcscore{
	Let $\I=(\candidates,\O,r)$ denote an election instance where $\O$ is a partitioned profile.
	A distinguished candidate $c \in \candidates$ is a possible winner (co-winner) in $\I$ if and only if it is a possible winner (co-winner) in $\I^c=(\candidates, \O^c,r)$. 
}

\begin{lemma}\label{lem:fixcscore}
	\lemmafixcscore
\end{lemma}

\subsection{$k$-approval}
Let $\I=(\candidates,\O,r)$ denote an election where $\O$ is a partitioned profile, and $r$ is the $k$-approval voting rule.
As a consequence of Lemma~\ref{lem:fixcscore}, when dealing with partitioned preferences, we may restrict our attention to extensions of $\O^c$ where $c$ is positioned at the top of its partition in every vote. Specifically, in every profile $\V^c$ that extends $\O^c$, candidate $c$ gets exactly $\smax(\O,c)$ points.
Now, consider any other candidate $c'\neq c$. 
If $\smax(\O,c')<\smax(\O,c)$ then  we have that
$s(\V^c,c')\leq \smax(\O,c')<\smax(\O,c)=s(\V^c,c)$.
Therefore, $c'$ cannot be a  winner (or co-winner) in any complete profile $\V^c \in \lin(\O^c)$.

Otherwise, if $\smax(\O,c')>\smax(\O,c)$ then $c$ can top $c'$ in $\V^c$ only if $c'$ is ranked in positions $k+1,\dots,m$ (i.e., receive $0$ points) in at least $\smax(\O,c')-\smax(\O,c)$ of the votes in which it could have received a point. Lemma~\ref{lem:cWinnerKApproval} below formalizes this condition. The proof is deferred to the appendix.

\def\kApprovalLemma{
		Let $\I=(\candidates,\O,r)$ be an election instance where $\O$ is a partitioned profile, and $r$ is the $k$-approval voting rule.
		Candidate $c$ is a possible co-winner in $\I$ if and only if there exists a complete profile $\V^c \in \lin(\O^c)$ where every candidate $c'\neq c$ is ranked in positions $\set{k+1,\dots,m}$ in at least $\left(\smax(\O,c')-\smax(\O,c)\right)$  of the votes in which $c'$ can receive a point (i.e., $\set{o_i \in \O \mid \smax(o_i,a)=1}$).	
}
\begin{lemma}\label{lem:cWinnerKApproval}
\kApprovalLemma
\end{lemma}
\subsubsection{Network Flow Algorithm}
Let $\I=(\candidates,\O,r)$ be an election where $\O$ is a partitioned profile, $r$ is the $k$-approval voting rule, and $c\in \candidates$ is a distinguished candidate.
We apply Lemma~\ref{lem:cWinnerKApproval} in a maximum network flow algorithm for deciding whether $c$ is a possible co-winner in $\I$.
We begin by describing the network and then prove the correctness of the algorithm.
\paragraph*{Network Description}
The network will contain the following sets of nodes:
\begin{enumerate}
	\item A source node $s$, and sink node $t$.
	\item Candidate nodes $V_{\candidates}$: all candidates $c_i\in \candidates$ for which $\smax(\O,c_i)>\smax(\O,c)$.
	\item Vote nodes $V_\O$: 
	For every vote $o_i\in \O$ where $o_i=(A^i_1\succ \dots \succ A^i_{q_i})$, the network will contain a \e{single} node $o_{i,j}$ where $A^i_j$ ($1\leq j \leq q_i$) is the partition containing the index $k$. For example, in the vote $o_i$ of Figure~\ref{fig:NetworkBuildEx}, the node $o_{i,2}$ represents the second partition $A_2^i$.
	\eat{
	That is:
	\[
	\sum_{l=1}^{j-1}|A_i^{l}|+1\leq k\leq\sum_{l=1}^j|A^i_l|
	\]
	}	
\end{enumerate}
The edges of the network:
\begin{enumerate}
	\item The set of edges $E_{s,\candidates}=\set{(s,c_i)\colon c_i \in V_{\candidates} }$. The capacity of edge $(s,c_i)$ is 	
	$u(s,c_i)=s_{max}(\O,c')-s_{max}(\O,c)$.	
	By construction, the capacity is strictly positive.
	\item A candidate node $a$ will have outgoing edges to all vote nodes $o_{i,j}\in V_\O$ in which it belongs to partition $A_j^i$ (in which it can lose a point). Formally:
	\[
	E_{V_{\candidates},V_{\O}}=\set{(a,o_{i,j}): a\in A^i_j}
	\]
	The capacity of every edge in $E_{V_{\candidates},V_{O}}$ is $1$.
	\item The set of edges $E_{V_\O,t}=\set{(o_{i,j},t)\colon o_{i,j}\in V_\O}$. 
	The capacity of every edge $(o_{i,j},t)$ is set to the number of positions in the partition $A_j^i$ whose corresponding score is $0$. Formally, $u(o_{i,j},t)=\sum^j_{l=1}|A_l^i|-k$.	
	For example, in the vote $o_i$ of Figure~\ref{fig:NetworkBuildEx}, the corresponding edge capacity is $u(o_{i,2},t)=\sum^2_{l=1}|A_l^i|-k$.
\end{enumerate}

\begin{figure}[h]
	\begin{center}
		\includegraphics[width=0.35\textwidth]{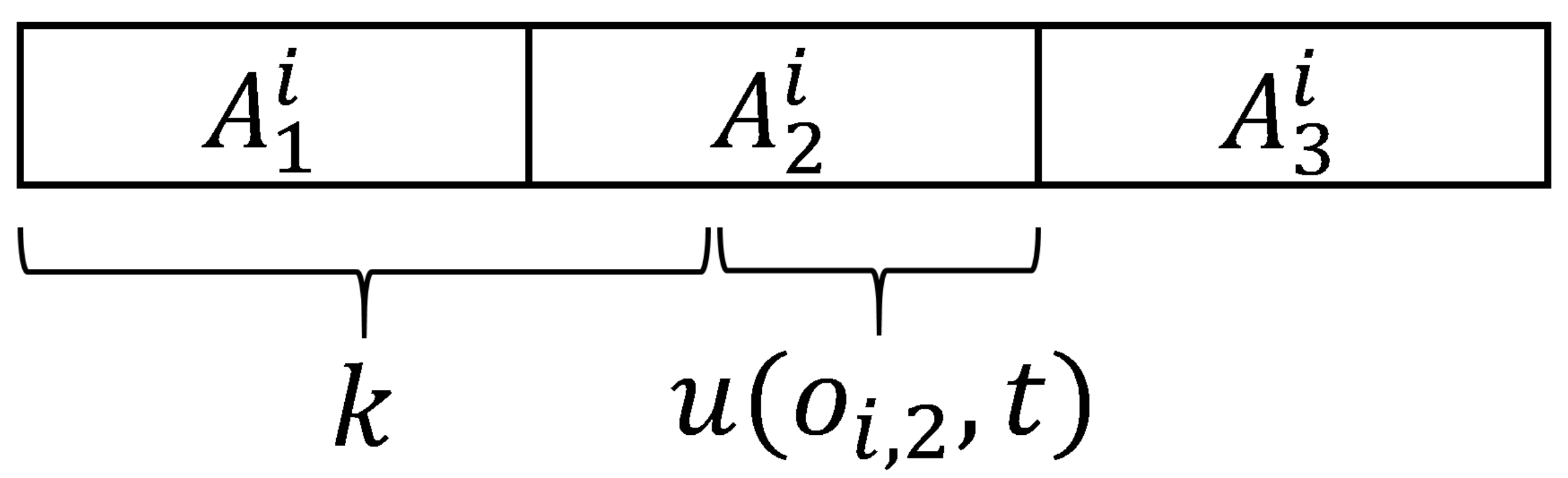}
	\end{center}
	\caption{A partitioned vote $o_i=\set{A_1^i \succ A_2^i \succ A_3^i}$}
	\label{fig:NetworkBuildEx}
\end{figure}

\begin{theorem}
	Let $\I=(\candidates,\O,r)$ be an election where $r$ is  $k$-approval and $\O$ is a partitioned profile.
	A distinguished candidate $c\in \candidates$ is a possible co-winner in $\I$ if and only if the maximum flow in the network is
	\begin{equation}\label{eq:maxFlow}
	\sum_{\substack{\set{a\in \candidates \mid  \smax(\O,a)> \smax(\O,c)}}}\left(\smax(\O,a)-\smax(\O,c)\right)	
	\end{equation}
\end{theorem}
\begin{proof}
	\textbf{The if direction}.\\
	Suppose that $c$ is a possible winner in $\I$.
	By Lemma~\ref{lem:cWinnerKApproval}, there exists a complete profile $\V^c \in \O^c$ such that every candidate $c'\neq c$ is ranked in positions $\set{k+1,\dots,m}$ in at least $(\smax(\O,c')-\smax(\O,c))$ of the votes $o_i$ in which $\smax(o_i,c')=1$.  That is, in $\V^c=(v_1,\dots,v_n)$, there exist $(\smax(\O,c')-\smax(\O,c))$ votes $v_i\in \V^c$ in which $\smax(o_i,c')=1$ while $s(v_i,c')=0$. Since $r$ is $k$-approval then in every such vote $v_i$, candidate $c'$ is ranked in a position strictly greater than $k$ (but smaller than the index corresponding to its partition in $v_i$). By the way we constructed the network, there exist at least $(\smax(\O,c')-\smax(\O,c))$ nodes $o_{i,l}\in V_\O$ for which there is a directed edge $(c',o_{i,l})$. Pushing a flow of $1$ on these edges, and repeating for every candidate $c'\neq c$ results in the required maximum flow.
	
	\textbf{The only if direction}\\
	So now, assume that we have a maximum network flow~\eqref{eq:maxFlow}, and we show how to construct a profile $\V^c \in \O^c$ in which $c$ is the winner.
	A maximum flow of~\eqref{eq:maxFlow} implies that every candidate node $c_i\in V_{\candidates}$ was able to push all of its incoming flow of $\left(\smax(\O,c_i)-\smax(\O,c)\right)$ to the vote nodes. 
	That is, there exist precisely $\left(\smax(\O,c')-\smax(\O,c)\right)$ nodes $o_{i,j}$ that received a unit of flow from $c'$. In each of the corresponding votes $o_i$, in which candidate $c'$ belongs to partition $A^i_j$, we position candidate $c'$ somewhere in the range of positions $\set{k+1,\dots,\sum_{l=1}^j|A^i_l|}$ where it receives a score of $0$.
	This is possible because given the maximum flow, and according to the capacities assigned to the edges from nodes $V_\O$ to $t$, we know that the number of candidates assigned to these positions in the vote $o_i$ does not exceed the capacity of $u(o_{i,j},t)$. Repeating this procedure for every candidate node $c'\neq c$, and placing the rest of the candidates in arbitrary positions, results in a complete ranking that abides to the conditions of Lemma~\ref{lem:cWinnerKApproval}, making $c$ a possible winner in $\I$.
\end{proof}

\begin{example}\label{ex:NetworkFlow}
	Let $\I=(\candidates,\O,r)$ be an election instance where $r$ is the $2$-approval voting rule, $\candidates=\set{a,b,c,d,e}$, and $\O$ is a partitioned profile defined as follows.
	\begin{center}
		\begin{tabular}{ c|lcl } 
			$o_1$ & $\set{b,c,d,e}$ & $\succ$ & $\set{a}$ \\ 
			$o_2$ & $\set{b,c,d}$ & $\succ$ &$ \set{a,e}$ \\ 
			$o_3$ & $\set{b,e}$ & $ \succ$ & $ \set{a,c,d}$ \\ 
			$o_4$ & $\set{b,d}$ & $\succ$ & $ \set{a,c,e}$ \\ 
			$o_5$ & $\set{c,d,e}$& $ \succ$ & $ \set{a,b}$ \\ 
			$o_6$ & $\set{c}$& $\succ$ & $ \set{a,b,d,e}$ 
		\end{tabular}
	\end{center}
	The table below presents the number of points each candidate $c'$  has to lose (with respect to $\smax(c',\O)$) so that $c$ is the winner.
	\begin{center}
		\begin{tabular}{  c | c   }
			Candidate & $\smax(\O,\cdot)-\smax(\O,c)+1$\\
			\hline
			$a$ & $0$ \\ \hline
			$b$ & $2$ \\ \hline
			$d$ & $2$ \\ \hline
			$e$ & $1$ \\ 
		\end{tabular}
	\end{center}
	The resulting network is presented in Figure~\ref{fig:NetworkEx}.
	The blue edges can carry a capacity of $1$. Bold edges represent a flow that takes up the capacity of the edge.
	The flow presented in the figure may correspond to one or more complete profiles $\V^c \in \O^c$ in which $c$ is the winner.
	\eat{
	Below is an example of such.
		\begin{center}
		\begin{tabular}{ c|c } 
			$v_1$ & $c \succ d \succ b \succ e \succ a$ \\ 
			$v_2$ & $c \succ b \succ d \succ a \succ e$ \\ 
			$v_3$ & $b \succ e \succ a \succ c \succ d$ \\ 
			$v_4$ & $b \succ d \succ a \succ c \succ e$ \\ 
			$v_5$ & $c \succ e \succ d \succ a \succ b$ \\ 
			$v_6$ & $c \succ a \succ d \succ b \succ e$ 
		\end{tabular}
	\end{center}
	One can easily check that $s(\V^c,c)=4$, $s(\V^c,a)=1$, $s(\V^c,b)=3$, $s(\V^c,d)=2$, and $s(\V^c,e)=2$ making $c$ the unique winner.
}
\end{example}

\begin{figure}[h]
	\begin{center}
		\includegraphics[width=0.4\textwidth]{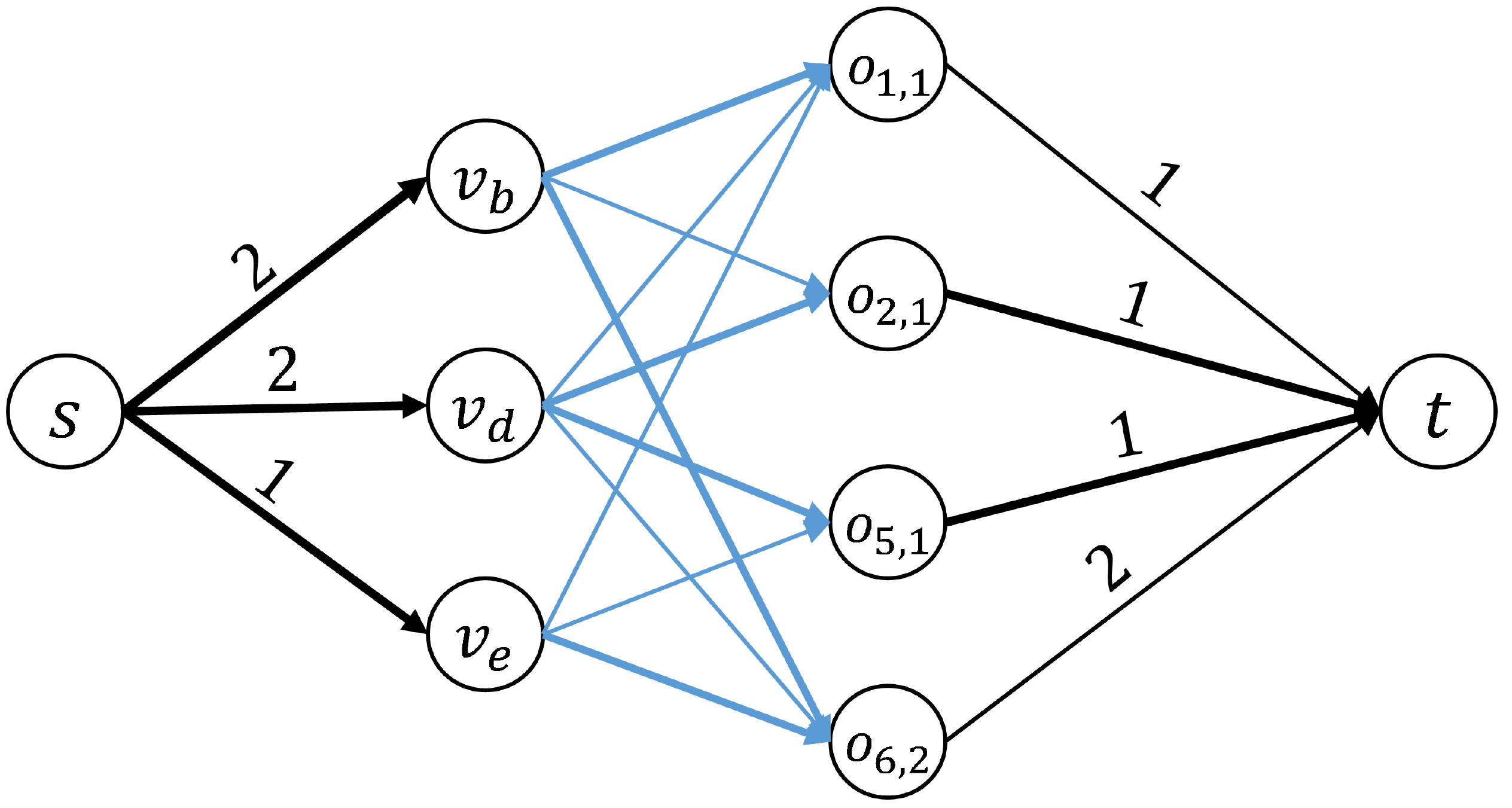}
	\end{center}
	\caption{The network and flow of Example~\ref{ex:NetworkFlow}. Bold edges indicate a flow taking up the full capacity of the edges.}
	\label{fig:NetworkEx}
\end{figure}

\subsection{The Positional Scoring rule $(2,1,1,\dots,1,0)$}
We now consider an election $\I=(\candidates,\O,r)$ where $r$ is the $(2,1,\dots,1,0)$ rule and $\O$ is a partitioned profile. As usual, $c\in \candidates$ is our distinguished candidate.
It has been shown that, in general, the $\PW$ problem for $(2,1,\dots,1,0)$ is NP-complete~\cite{DBLP:journals/corr/abs-1108-4436}. We show that the network flow algorithm of the previous section solves this problem in polynomial time if $\O$ is a partitioned profile.

Let $o_i \in \O$ denote a partitioned vote and $c'\neq c$ a candidate with a maximum score $\smax(o_i,c')$ in $o_i$.
If $o_i$ has two or more partitions then in any extension $v_i \in o_i$ exactly one of the following can occur: (1) $\smax(o_i,c')=s(v_i,c')=2$, (2) $\smax(o_i,c')=2$ and $s(v_i,c')=1$, (3) $\smax(o_i,c')=s(v_i,c')=1$ or (4) $\smax(o_i,c')=1$ and $s(v_i,c')=0$.
In all of these options, candidate $c'$ can lose either 0 or 1 points in $o_i$. Formally, for any candidate $c'\neq c$, and any partitioned vote $o_i$ with at least two partitions we have that $\left(\smax(o_i,c')-s(v_i,v')\right) \in \set{0,1}$.

Now, let us assume that $o_i \in \O$ is a partitioned preference with a single partition. That is, $o_i$ contains no precedence constraints. In this case, by Lemma~\ref{lem:fixcscore}, we can assume that any complete profile $\V$ in which $c$ wins (or co-wins), is an extension of $\O^c$. In particular, this means that we may assume that in $v_i \in \V$, candidate $c$ is ranked in the topmost position and thus receives two points (i.e., $s(v_i,c)=2$). This, in turn, means that for any other candidate $c' \neq c$ exactly one of the following can occur: (1) $\smax(o_i,c')=s(v_i,c')=1$ (2) $\smax(o_i,c')=1$ and $s(v_i,c')=0$. As in the previous case, candidate $c'$ can lose either 0 or 1 points in $o_i$. Formally, $\smax(o_i,c')-s(v_i,v') \in \set{0,1}$.

Now that we have established that every candidate can ``lose'' at most one point in every vote, we can apply the network flow algorithm of the previous section.

\section{Hardness}\label{sec:hardness}
Let $r=(\vec{\balpha}_m)_{m \in \mathbb{N}^+}$ be a pure, positional scoring rule. 
From this point on we assume that $r$ produces a scoring vector with at least $3$ distinct values (the case of $2$-valued scoring rules was considered in the previous section).

\begin{definition}(\cite{dey_et_al:LIPIcs:2017:8135})
	We say that a voting rule $r=(\vec{\balpha}_m)_{m \in \mathbb{N}^+}$ is \e{differentiating} if there exists some constant $n_0 \in \mathbb{N}^+$ such that for all $m \geq n_0$ the score vector $\vec{\balpha}_m$ contains two positions $i,j\geq 2$ where $j > i+1$ such that $(\alpha_{j}-\alpha_{j-1}) > (\alpha_{i}-\alpha_{i-1})$.
\end{definition}
Dey and Misra~\cite{dey_et_al:LIPIcs:2017:8135}
have shown that the possible winner problem is NP-complete for all differentiating scoring rules. The proof (Theorem 6) relies only on partitioned preferences, implying hardness of the $\PW$ problem for differentiating scoring rules with partitioned profiles. Therefore, we restrict our attention to non-differentiating scoring rules. Formally, for every scoring vector $\vec{\balpha}_m$, and for every pair of consecutive values $\alpha_i,\alpha_{i+1}$, we have that $\alpha_{i+1}-\alpha_i \leq 1$.

A common strategy in proving hardness for the PW problem is to construct a profile $\Q$, consisting of a set of linear orders, 
that enables determining the score of every candidate in $\candidates$ according to the requirements dictated by the reductions~\cite{DBLP:journals/jcss/BetzlerD10,dey_et_al:LIPIcs:2017:8135,DBLP:conf/atal/BaumeisterRR11,DBLP:journals/corr/abs-1108-4436}.
\eat{such that the score of every candidate in $\candidates$ can be determined based on the scores that are required by the reductions~\cite{DBLP:journals/jcss/BetzlerD10,dey_et_al:LIPIcs:2017:8135,DBLP:conf/atal/BaumeisterRR11,DBLP:journals/corr/abs-1108-4436}.} Once such a set is constructed, the profile is enhanced with a set of partial votes $\P$, where the maximum scores of the candidates are restricted according to the linear votes in $\Q$. Lemma~\ref{lem:linearVotes} below~\cite{dey_et_al:LIPIcs:2017:8135} states that such a profile $\Q$ can be constructed in polynomial time.

\begin{lemma}[\cite{dey_et_al:LIPIcs:2017:8135}] \label{lem:linearVotes}
	Let $\candidates=\set{c_1,\dots,c_{m}}\cup D$, $(|D|>0)$ be a set of candidates, and $\vec{\balpha}$ a scoring vector of length $|\candidates|$. Then for every integer vector $\bX=(X_1,\dots,X_m)\in \mathbb{Z}^m$, there exists a $\lambda \in \mathbb{N}$ and a voting profile $\Q$ such that $s(c_i,\Q)=\lambda+X_i$ for all $1\leq i \leq m$, and $s(d,\Q)<\lambda$ for all $d\in D$. Moreover, the number of votes in $\Q$ is polynomial in $|\candidates|\cdot\sum_{i=1}^m|X_i|$.
\end{lemma}

\eat{
Some of our NP-hardness proofs rely on reductions from the NP-complete \algname{Exact Cover By 3-Sets} ($\ex3C$) problem~\cite{Garey:1990:CIG:574848}, while another theorem relies on a reduction from the NP-complete problem \algname{3-Dimensional Matching} ($\threeDM$) ~\cite{Garey:1990:CIG:574848}. 
The $\ex3C$ problem is defined as follows. Given a set of elements $E=\set{e_1,\dots,e_{3M}}$, a family of subsets $\S=\set{S_1,\dots,S_t}$ with $|S_i|=3$ and $S_i\subseteq E$ for $1\leq i\leq t$, it asks whether there is a subset $\S'\subset \S$ such that $|\S'|=M$, and for every element $e_j \in E$
there is exactly one set $S_i\in \S'$ such that $e_j \in S_i$.
}

Our NP-hardness proofs rely on reductions from the NP-complete 3-Dimensional-Matching problem ($\threeDM$)~\cite{Garey:1990:CIG:574848}. 
The $\threeDM$ problem is defined as follows.
We are given three disjoint sets $\X$, $\Y$, and $\Z$ each containing exactly $M$ elements, and a set $\S\subseteq \X \times \Y \times \Z$ of triples. We wish to know whether there is a subset $\S' \subset \S$ of $M$ disjoint triples that covers all elements of $\X \cup \Y \cup \Z$.

In some of our theorems, we will need functions that map each instance $\I$ of $\threeDM$  to a natural number, and in some sense behave like a polynomial. For this sake, we call
\[
f:\set{\I \mid \I \text{ is an instance of }\threeDM }\mapsto \mathbb{N}
\]
a \e{poly-type function for $\threeDM$}~\cite{DBLP:journals/jcss/BetzlerD10} if the function value $f(\I)$ is bounded by a polynomial in $|\I|$ for every input instance $\I$ of $\threeDM$.

\def\3DMReduction{
	A $\threeDM$ instance $\I$ can be reduced to a $\PW$ instance for a scoring rule which produces a size-$m$ scoring vector that fulfills the following. There is an $i \geq 2$ such that $\alpha_{i+k} > \alpha_{i+k-1}=\cdots=\alpha_{i}>0$ with $k\geq 1$, and $m-k=f(\I)$. A suitable poly-type function $f$ for $\threeDM$ can be computed in polynomial time. 
}
\begin{lemma}\label{lem:3DMReduction}
	\3DMReduction
\end{lemma}

\begin{proof}
Let $a$ denote the value that occupies positions $i,\dots,i+k-1$ in $\vec{\balpha}_m$.
By the previous discussion, and since $r$ is non-differentiating, the scoring vector $\vec{\balpha}_m$ contains three indexes $i$, $i-1$, and $i+k$ such that $\alpha_{i-1}=a-1$, $\alpha_{i}=a$, $\alpha_{i+k} =a+1$.
Schematically:
\begin{equation}\label{eq:schmaticVec}
\vec{\balpha}_m=\left(\dots,\overset{\overset{i+k}{\big\downarrow}}{a+1}, \underbrace{a,\dots,\overset{\overset{i}{\big\downarrow}}{a}}_{k},\overset{\overset{i-1}{\big\downarrow}}{a-1},\dots\right)
\end{equation}

Let $\I=(E,\S)$ denote a $\threeDM$ instance where $E=\X\cup\Y\cup\Z$. The set $\candidates$ of candidates is defined by $\candidates\asn \set{c} \cup E \cup H \cup D$ 
where $c$ denotes the distinguished candidate, $E$ the set of candidates that represent the elements of the $\threeDM$ instance, and $H$ and $D$ contain disjoint candidates such that the following hold.
We define $H=\cup_{\s \in \S}H_\s$ where the sets $H_\s$ are pairwise disjoint, and $|H_\s|=k-1$ for all $\s\in \S$. The sets $H_\s$ will be used for ``padding'' some positions relevant to the construction. The set $D$ contains $m-|E|-|H|-1$ candidates needed to pad irrelevant positions. 
We set $f(\I)=|\candidates\setminus D|-k=|E|+(|\S|-1)(k-1)-1$. Recall that $f(\I)=m-k$. Intuitively, this means that the portion of the scoring vector $\vec{\balpha}_m$, 
occupied by values different from $\alpha_i$, is large enough to contain all elements besides one of the sets $H_\s$.

For every triple $\s=(x,y,z) \in \S$  
let $\candidates_\s \subset (\candidates\setminus \set{x,y,z})$ such that $|\candidates_\s|=i-2$ (see~\eqref{eq:schmaticVec}).
We construct the following linear vote $\vote{v}_\s$.
\[
\vote{v}_\s=\overrightarrow{\left(\candidates\setminus \candidates_\s\right)}\succ x \succ y \succ \overrightarrow{(H_\s)}\succ z \succ \overrightarrow{\left(\candidates_\s\right)}
\]
where $\overrightarrow{\left(\candidates\setminus \candidates_\s\right)}$, $\overrightarrow{(\candidates_\s)}$, and $\overrightarrow{(H_\s)}$ are arbitrary complete orders over the candidate sets $\candidates\setminus \candidates_\s$, $\candidates_\s$, and $H_\s$ respectively.
Using $\vote{v}_\s$ we define the partial \emph{partitioned} vote $\vote{v}'_\s$ as follows.
\[
\vote{v}'_\s=\overrightarrow{\left(\candidates\setminus \candidates_\s\right)}\succ \left(x \cup y \cup H_\s \cup z\right) \succ \overrightarrow{\left(\candidates_\s\right)}
\]
Note that this implies that in any extension of $\vote{v}'_\s$, items $H_\s\cup \s$ will occupy the positions in the range $\set{i-1,\dots,i+k}$.

We denote by $\P=\cup_{\s \in \S}\vote{v}_\s$ and $\P'=\cup_{\s \in \S}\vote{v}'_\s$. 
By Lemma~\ref{lem:linearVotes} there exists a set of linear votes $\Q$, of size polynomial in $m$, where the scores of the candidates in the combined profile $\P \cup \Q$ are as follows:
\begin{align*}
s_{\P\cup\Q}(x)&=s_{\P\cup \Q}(c)+2&\forall x \in \X  \\
s_{\P\cup\Q}(y)&=s_{\P\cup \Q}(c)-1&\forall y \in \Y  \\
s_{\P\cup\Q}(z)&=s_{\P\cup \Q}(c)-1&\forall z \in \Z  \\
s_{\P\cup\Q}(h)&=s_{\P\cup \Q}(c)&\forall h \in H  \\
s_{\P\cup\Q}(d)&<s_{\P\cup\Q}(c)&\forall d \in D
\end{align*}
We observe that the score of $c$ is the same in any extension of $\P'\cup Q$ and is identical to its score in $\P \cup Q$.
We define the instance $\W$ of $\PW$ to be $(\candidates, \P'\cup \Q, r)$, and proceed with the reduction.

In the forward direction, suppose that $\I$ is a $\mathbf{Yes}$ instance of $\threeDM$. Then, there exists a collection of $M$ disjoint sets $\S'\subset \S$ in $\S$ such that $\cup_{\s \in \S}\s=\X\cup \Y \cup \Z$. For every $\s \in \S$ we extend the partial vote $\vote{v}'_\s$ to $\ol{\vote{v}'_\s}$ as follows.
\begin{equation*}
\ol{\vote{v}'_\s}=
\begin{cases}
\overrightarrow{\left(\candidates\setminus \candidates_\s\right)} \succ y \succ \overrightarrow{(H_\s)}\succ z \succ x \succ \overrightarrow{(\candidates_\s)}& \s \in \S'\\
\overrightarrow{\left(\candidates\setminus \candidates_\s\right)} \succ x \succ y \succ \overrightarrow{(H_\s)} \succ z \succ \overrightarrow{(\candidates_\s)}& \s \not\in \S'\\
\end{cases}
\end{equation*}
where, again, $\overrightarrow{(H_\s)}$ is an arbitrary complete order over the candidates $H_\s$.
We consider the extension of $\P'$ to $\ol{\P'}=\cup_{\s \in \S}\ol{\vote{v}'_\s}$. 
We claim that $c$ is a co-winner in the profile $\ol{\P'}\cup \Q$ because:
\begin{enumerate}
	\item  For all $x \in \X$:  $s_{\ol{\P'}\cup \Q}(x)=s_{\P\cup\Q}(x)-2=s_{\P\cup\Q}(c)$
	\item For all $y \in \Y$: $s_{\ol{\P'}\cup \Q}(y)=s_{\P\cup\Q}(y)+1=s_{\P\cup\Q}(c)$
	\item For all $z \in \Z$:
	$s_{\ol{\P'}\cup \Q}(z)=s_{\P\cup\Q}(z)+1=s_{\P\cup\Q}(c)$. 
	\item For all $h \in H_\s$:
	$s_{\ol{\P'}\cup \Q}(h)=s_{\P\cup\Q}(h)=s_{\P\cup\Q}(c)$. 
	\item For all $d \in D$:
	$s_{\ol{\P'}\cup \Q}(d)=s_{\P\cup\Q}(d)<s_{\P\cup\Q}(c)$. 
\end{enumerate}

For the reverse direction, suppose that the $\PW$ instance $\W$ is a $\mathbf{Yes}$ instance. Then there exists an extension of the set of partial, partitioned votes $\P'$ to a set of complete votes $\ol{\P'}$ such that $c$ is a co-winner in $\ol{\P'}\cup \Q$. We refer to the extension of $\vote{v}'_\s \in \P'$ as $\ol{\vote{v}'_\s} \in \ol{\P'}$.

We recall that the score of $c$ is the same in any extension of $\P'\cup \Q$ and is identical to its fixed score in $\P \cup \Q$.
We first claim that for every $\ol{\vote{v}'_\s} \in \ol{\P'}$ in which $x$ is not in position $i+k$ (see~\eqref{eq:schmaticVec}), then $y$ occupies this position. Indeed, if $z$ occupied this position then its score would increase by at least $2$ (i.e., compared to the vote $\vote{v}_\s$). Since $z$ cannot lose points in any of the votes $\vote{v}'_\s$ this would mean that $s_{\ol{\P'}\cup \Q}(z)\geq s_{\P\cup \Q}(z)+2= s_{\P\cup \Q}(c)+1$. But then we arrive at a contradiction that $c$ is a co-winner.
Likewise, if some $h\in H_\s$ occupied this position then its score would increase by at least $1$ because, according to the construction, the position of $h$ is fixed in the rest of the votes, and thus, it cannot lose points in any other vote in $\P'$. Then, $s_{\ol{\P'}\cup \Q}(h)\geq s_{\P\cup \Q}(h)+1=s_{\P\cup \Q}(c)+1$ contradicting the assumption that $c$ is a co-winner. Finally, candidates in $D$ cannot occupy position $i+k$ in any extension of $\P'$.

We now claim that, for every $x \in \X$, there exists exactly one triple $\s \in \S$ such that $x$ is not in position $i+k$ in $\ol{\vote{v}'_\s}$.
Otherwise, by the previous claim, there must be at least two votes in which position $i+k$ is occupied by candidates from $\Y$. 
Since our profile contains more than $M$ votes (i.e., $|\S|>M$)), then by the pigeon-hole principle there exists a candidate $y'\in \Y$ that appears in position $i+k$ at least twice. But then, the overall score of candidate $y'$ must have increased by strictly more than $1$ point.
However, in such a scenario, the score of $y'$ will be strictly more than the score of $c$ contradicting the fact that $c$ is a co-winner in $\ol{\P'}\cup \Q$. 

Now, the claim follows from the observation that every $x\in \X$ must lose $2$ points in order for $c$ to co-win. 
From the claim that there is exactly one vote $\ol{\vote{v}'_s}$ in which $x$ does not occupy position $i+k$ for every $x \in \X$, and since $|\X|=M$, we have that $\ol{\P'}$ contains precisely $M$ votes corresponding to triples $\S'\subset \S$ in which $x$ does not occupy position $i+k$. Furthermore, for every $\s \in \S'$ it must be the case that $x$ loses two points and thus $\ol{\vote{v}'_\s}=\overrightarrow{\left(\candidates\setminus \candidates_\s\right)} \succ y \succ \overrightarrow{(H_\s)}\succ z \succ x \succ \overrightarrow{(\candidates_\s)} $.

We now show that $\cup_{\s \in \S'}=\X \cup \Y \cup \Z$. It is clear that $\cup_{\s \in \S'}\subseteq \X \cup \Y \cup \Z$, so we show that $\cup_{\s \in \S'} \supseteq \X \cup \Y \cup \Z$. Assume the contrary. Then there is a candidate $u \in \Y \cup \Z$ that does not belong to $\cup_{\s \in \S'}$. If $u \in \Y$, then by the claim that in every vote in which $x$ is not in position $i+k$, there is some candidate $y\in \Y$ in this position, and that $|\S'|=M$, there must be some candidate $y' \in \Y$ that appears at least twice in this position. But then $s_{\ol{\P'}\cup \Q}(y') \geq s_{\P\cup \Q}(y')+2=s_{\P\cup \Q}(y')+1$, and we arrive at a contradiction that $c$ is a co-winner.
If $u \in \Z$ then by the same reasoning there is some candidate $z'\in Z$ that appears at least twice in position $i$ (see~\eqref{eq:schmaticVec}). But then $s_{\ol{\P'}\cup \Q}(z') \geq s_{\P\cup \Q}(z')+2=s_{\P\cup \Q}(z')+1$, and we arrive at a contradiction that $c$ is a co-winner. 
\end{proof}

\def\3DMReduction2{
	A $\threeDM$ instance $\I$ can be reduced to a $\PW$ instance for a scoring rule which produces a size-$m$ scoring vector that fulfills the following. There is an $i \geq 1$ such that $\alpha_{i+l-1}=\cdots=\alpha_{i}$ with $l=f(\I)$ and there exists an index $j$ such that $\alpha_j \geq \alpha_i+2$ or $\alpha_j \leq \alpha_i-2$.
	A suitable poly-type function $f$ for $\threeDM$ can be computed in polynomial time. 
}
\begin{lemma}\label{lem:3DMReduction2}
	\3DMReduction2
\end{lemma}
\begin{proof}
	We assume, without loss of generality, that $\alpha_j \geq \alpha_i+2$, the other case (i.e., $\alpha_j \leq \alpha_i-2$) is symmetrical. Schematically, the scoring vector $\vec{\balpha}_m$ is:
	\begin{equation}\label{eq:schmaticVec3DM}
	\left(\dots,\overset{\overset{i+l+k}{\big\downarrow}}{\alpha_{i+2}},\underbrace{\alpha_{i+1},\dots,\overset{\overset{i+l}{\big\downarrow}}{\alpha_{i+1}}}_{k},\underbrace{\alpha_i,\dots,\overset{\overset{i}{\big\downarrow}}{\alpha_i}}_{l=3M-2},\cdots\right)
	\end{equation}
	We define $l=f(\I)=3M-2$ the number of positions carrying the value $\alpha_i$, and
 	let $k$ denote the number of positions carrying the value $\alpha_{i+1}$ (see~\eqref{eq:schmaticVec3DM}).
	
	Let $\I=(E,\S)$ denote a $\threeDM$ instance where $E=\X\cup\Y\cup\Z$. The set $\candidates$ of candidates is defined by $\candidates\asn \set{c} \cup E \cup H \cup D$ 
	where $c$ denotes the distinguished candidate, $E$ the set of candidates that represent the elements of the $\threeDM$ instance, and $H$ and $D$ contain disjoint candidates such that the following hold. 	
	The set $H$ contains dummy candidates that pad all but one of the positions occupied by value $\alpha_{i+1}$. That is, $H=\set{h_1,\dots,h_{k-1}}$. The dummy candidates $D$ pad the positions that are irrelevant to the construction, such that $|D|=m-|E|-|H|-1$.
	
	We build a partial profile $\P\cup Q$ that consists of a set of complete votes $\Q$, and a set of partial, partitioned votes $\P$. The set of partitioned votes $\P$ is defined as follows. For every triple $\s=(x,y,z)\in \S$ we let $\candidates_{\s} \subseteq D\cup\set{c}$ such that $|\candidates_{\s}|=i-1$. The partitioned vote $\vote{p}_\s$ is defined as follows. 
	\[
	\vote{p}_\s=\overrightarrow{(\candidates \setminus \candidates_{\s})} \succ \left(\s\cup H\right) \succ \left(E\setminus \s \right)\succ \overrightarrow{(\candidates_{\s})}
	\]
	where $\overrightarrow{(\candidates \setminus \candidates_{\s})}$, and $\overrightarrow{(\candidates_{\s})}$ denote arbitrary complete orders over the sets of candidates $(\candidates \setminus \candidates_{\s})$ and $\candidates_{\s}$ respectively. We note that $|\s \cup H|=k+2$ and this set occupies positions $i+l-1,\dots,i+l+k$ in $\vote{p}_\s$, and that $|E \setminus \s|=3M-3$, occupying positions $i,\dots,i+l-2$ in  $\vote{p}_\s$.

	Since the position of $c$ is fixed in $\Q\cup \P$ it has the same score in any extension $\Q\cup\overline{\P}$ of $\Q \cup \P$. We denote this score by $s_{\Q \cup \P}(c)$.
	For any candidate $c' \in \candidates \setminus \set{c}$, the  \e{maximum partial  score}  $s_\P^{max}(c')$ is the maximum number of points $c'$ may make in $\P$ without beating $c$ in $\P \cup \Q$. Since the score of $c$ is fixed in $\Q \cup \P$,  then in order for $c$ to win the election we must have that $s_\P^{max}(c')=s_{\Q \cup \P}(c)-s_{\Q}(c')$.

	According to Lemma~\ref{lem:linearVotes}, we can set the score of any candidate $c'\in \candidates \setminus \set{c}$ in $\Q$ such that its maximum partial score $s_\P^{max}(c')$ is as follows.	
	\begin{enumerate}
		\item  For all $h \in H$: 
		 $s^{max}_{\P}(h)=|\S|\cdot \alpha_{i+1}$
		
		\item For all $x\in X$:
		\[
			s^{max}_{\P}(x)=\left( (n_x-1)\cdot \alpha_{i+2} +(|\S|-n_x+1)\alpha_i\right)
		\]
		\item  For all $y \in \Y$: 
		\[
				s^{max}_{\P}(y)=\left( (n_y-1)\cdot \alpha_{i+1} + \alpha_{i+2} +(|\S|-n_y)\alpha_i\right)
		\]
		\item  For all $z \in \Z$: 
		\[
			s^{max}_{\P}(z)=\left( \alpha_{i+1} +(|\S|-1)\alpha_i\right)
		\]
		\item For all $d \in D$:
		$s^{max}_{\P}(d)> |\S|\cdot \alpha_m$.
	\end{enumerate}
where, for any item $e\in E$, $n_e$ denotes the number of occurrences of $e$ in the set of triples $\S$.

In the forward direction, suppose that $\I=(E,\S)$ is a $\mathbf{Yes}$ instance of $\threeDM$. Then, there exists a collection of $M$ sets $\S'\subset \S$ in $\S$ such that $\cup_{\s \in \S}\s=\X\cup \Y \cup \Z$. For every $\s \in \S'$ we extend the partial vote $\vote{p}_\s$ to $\ol{\vote{p}_\s}$ as follows.
\begin{equation*}
\ol{\vote{p}}_\s=
\begin{cases}
\overrightarrow{\left(\candidates\setminus \candidates_\s\right)} \succ y \succ \overrightarrow{(H)}\succ z \succ x \succ \overrightarrow{\left(E\setminus \s \right)} \succ \overrightarrow{(\candidates_\s)}& \s \in \S'\\
\overrightarrow{\left(\candidates\setminus \candidates_\s\right)} \succ x \succ y \succ \overrightarrow{(H)} \succ z \succ \overrightarrow{\left(E\setminus \s \right)} \succ \overrightarrow{(\candidates_\s)} & \s \not\in \S'\\
\end{cases}
\end{equation*}
We consider the extension of $\P$ to $\overline{\P}=\cup_{\s \in \S}\overline{\vote{p}_\s}$. We claim that $c$ is a co-winner in $\Q \cup \overline{\P}$ because:
\begin{enumerate}
	\item For all $h\in H$:
	$s_{\Q \cup \overline{\P}}(h)=|\S|\cdot \alpha_{i+1}$
	\item For all $x\in \X$:
	 $$s_{\Q \cup \overline{\P}}(x)=(n_x-1)\alpha_{i+2}+\alpha_i+(|\S|-n_x)\cdot \alpha_i$$
	 \item For all $y\in \Y$:
	 $$s_{\Q \cup \overline{\P}}(y)=\alpha_{i+2}+(n_y-1)\alpha_{i+1}+(|\S|-n_y)\cdot \alpha_i$$
	 \item For all $z\in \Z$:
	 $$s_{\Q \cup \overline{\P}}(z)=\alpha_{i+1}+(n_z-1)\alpha_{i}+(|\S|-n_z)\cdot \alpha_i$$
	 \item For all $d\in D$:
	 $s_{\Q \cup \overline{\P}}(d)\leq|\S|\alpha_m$
\end{enumerate}
That is, items $E \cup H$ precisely reach their maximum partial scores in $\Q \cup \overline{\P}$, thereby making $c$ a co-winner.

For the reverse direction, we verify a property of the construction called \e{tightness}~\cite{DBLP:journals/jcss/BetzlerD10} that is crucial to the correctness: In total, the score of all positions that must be filled in the $|\S|$ votes equals the sum of the maximum partial scores of all candidates. Once we establish tightness, it follows that a candidate $c'\in \candidates \setminus \set{c}$ cannot earn less than $s_{\P}^{max}(c')$ points since otherwise there must be another candidate $c''\in \candidates \setminus \set{c}$ that makes more than $s_{\P}^{max}(c'')$ points, and thus beats $c$. We now establish tightness with regard to the positions relevant to the construction $\set{i,\dots,i+l+k}$.
We have a total of $|\S|$ votes, and the candidates of $D\cup \set{c}$ are fixed at positions $\set{1,\dots,i-1}$, and $\set{i+l+k+1,\dots,m}$. Therefore, the total number of points for the remaining candidates $H\cup E$ is:
\begin{equation}\label{eq:scorePositions}|\S|\left(\alpha_{i+2}+(|H|+1)\cdot \alpha_{i+1} +\alpha_i + (3M-3)\alpha_i\right)
\end{equation}

We now consider the sum of maximum partial scores of candidates $H \cup E$. For the candidates of $H$:
\[
\sum_{h\in H}s_\P^{max}(h)=\sum_{h\in H}|\S|\cdot\alpha_{i+1}=|H|\cdot|\S|\cdot \alpha_{i+1}
\]
Now, we look at the sum of maximum partial scores of candidates $\X$. For this calculation note that $\sum_{x \in \X}n_x=|\S|$.
\begin{align*}
&\sum_{x\in \X}s_\P^{max}(x)=\sum_{x\in \X}\left((n_x-1)\alpha_{i+2}+(|\S|-n_x+1)\cdot \alpha_i\right)\\
&\underbrace{=}_{|\X|=M}M(|\S|+1)\alpha_i-M\alpha_{i+2}+\sum_{x\in \X}n_x\cdot \alpha_{i+2}-\sum_{x\in \X}n_x\alpha_i\\
&\underbrace{=}_{\sum_{x\in \X}n_x=|\S|}M(|\S|+1)\alpha_i-M\alpha_{i+2}+|\S|\alpha_{i+2}-|\S|\alpha_i\\
&=(|\S|-M)\alpha_{i+2}+M\alpha_i+\alpha_i|\S|(M-1)
\end{align*}
Now, we look at the sum of maximum partial scores of candidates $\Y$.
\begin{align*}
&\sum_{y\in \Y}s_\P^{max}(y)=\sum_{y\in \Y}\left( (n_y-1)\cdot \alpha_{i+1} + \alpha_{i+2} +(|\S|-n_y)\alpha_i\right)\\
&=\sum_{y \in \Y}n_y\cdot \alpha_{i+1}-\sum_{y \in \Y}n_y\cdot \alpha_i -M\alpha_{i+1}+M\alpha_{i+2}+M|\S|\alpha_i\\
&\underbrace{=}_{\sum_{y\in \Y}n_y=|\S|}(|\S|-M)\alpha_{i+1}+M\alpha_{i+2}+\alpha_i|\S|(M-1)
\end{align*}
Finally, we look at the sum of maximum partial scores of candidates $\Z$.
\begin{align*}
&\sum_{z\in \Z}s_\P^{max}(z)=\sum_{z\in \Z}\left( \alpha_{i+1} +(|\S|-1)\alpha_i\right)\\
&=M\alpha_{i+1}+M\alpha_i(|\S|-1)\\
&=M\alpha_{i+1}+\alpha_i|\S|(M-1)+|\S|\alpha_i-M\alpha_i\\
&=(|\S|-M)\alpha_i+M\alpha_{i+1}+\alpha_i|\S|(M-1)
\end{align*}
Adding up $\sum_{h\in H}s_\P^{max}(h)+\sum_{e\in E}s_\P^{max}(e)$ we get exactly the total score of all positions that must be filled (see~\eqref{eq:scorePositions}) with items $H\cup E$. Thus, tightness follows.

Now that we have established tightness,  suppose that the $\PW$ instance $\W=(\candidates,\P,r)$ is a $\mathbf{Yes}$ instance. Then there exists an extension of the set of partial, partitioned votes $\P$ to a set of complete votes $\ol{\P}$ such that $c$ is a co-winner in $\ol{\P}\cup \Q$. We refer to the extension of $\vote{p}_\s \in \P$ as $\ol{\vote{p}_\s} \in \ol{\P}$.

Before we continue we will require the following claim. For any $\s\in \S$ the item that occupies position $i+l+k$ in $\ol{\vote{p}_\s}\in \ol{P}$ is in $\X \cup \Y$.
The proof is as follows. Since $c$ is a co-winner in $\ol{P}$ then, by tightness, the total score of the items in $\X \cup \Y$ in $\ol{P}$ is:
\begin{align}
&\sum_{x \in \X}s^{max}_{\P}(x)+\sum_{y \in \Y}s^{max}_{\P}(y)= \nonumber
\\ 
&|\S|\alpha_{i+2}+M\alpha_i+(|\S|-M)\alpha_{i+1}+|\S|(2M-2)\alpha_i \label{eq:claim}
\end{align}
Since each vote provides the candidates of $\X \cup \Y$ with one of $\alpha_i,\alpha_{i+1}$, or $\alpha_{i+2}$ points,
then even a single vote, in which position  $i+l+k$ (scoring $\alpha_{i+2}$ points) is not occupied by an item in $\X \cup \Y$,
will result in a violation of tightness, because the sum will be strictly less than~\eqref{eq:claim} \footnote{For the case in which there is an index $j$ such that $\alpha_j= \alpha_i-2$, we denote by $\alpha_{i-2}=\alpha_i-2$, and the resulting sum would be $|\S|\alpha_{i-2}+M\alpha_i +(|\S|-M)\alpha_{i-1}+|\S|\alpha_i(2M-2)$. Then, we would reason that in every vote, the item that occupies the single position scoring $\alpha_{i-2}$ points, is in $\X \cup \Y$. Otherwise, some item $e$ in $\X\cup\Y$ scores more than $s^{max}_{\P}(e)$ points in $\ol{\P}$, and $c$ cannot be a co-winner.}. 

From the previous claim it follows that no item $h \in H$ can occupy position $i+l-1$ (scoring $\alpha_i$ points), because then, by tightness, this item would have to occupy position $i+l+k$ in at least one of the votes, contradicting the previous claim. Therefore, in all votes $\ol{\vote{p}_\s} \in \ol{\P}$ the items of $H$ occupy positions $i+l,\dots,i+l+k-1$ scoring $\alpha_{i+1}$ points.

We observe that for any triple $\s=(x,y,z)$, item $y$ must score either $\alpha_{i+1}$ or $\alpha_{i+2}$ in the extension $\ol{\vote{p}_\s}$. Otherwise, we immediately arrive at a violation of tightness with regard to $s^{max}_{\P}(y)$.
From this we conclude that for any triple $\s=(x,y,z)$, either item $x$ or $z$ occupy position $i+l-1$ (scoring $\alpha_i$ points) in $\ol{\vote{p}_\s}$.
Furthermore, every item $y \in \Y$ scores $\alpha_{i+2}$ points exactly once in $\ol{P}$. More than once, and $y$ scores more than $c$, and less would result in violation of tightness. 

By tightness, there exist precisely $M$ votes corresponding to $M$ triples $\S' \subseteq \S$ where item $z$ scores $\alpha_{i+1}$ points, freeing position $i+l-1$. As previously argued, only item $x$ can take this position, freeing position $i+l+k$ which can only be occupied by item $y$.
We now show that $\S'$ covers all items $\X \cup \Y \cup \Z$. Assume, by contradiction, that there is an item $x\in \X$ that is not covered by $\S'$. By the pigeon hole principle, there must be some item $x'\in \X$ that occupies position $i+l-1$, scoring $\alpha_i$ points, more than once. But then, we immediately arrive to a contradiction to tightness with regard to $s^{max}_{\P}(x')$ because $x'$ can score $\alpha_i$ exactly once for some triple in which it appears. If there is some item $y \in \Y$ that is not covered by $\S'$, then there must be some other item $y' \in \Y$ that scores $\alpha_{i+2}$ points more than once, making it the winner.
Overall, if $c$ is a co-winner in $\ol{\P}$ then the set of triples $\S'$, where $|\S'|=M$ cover all items $\X\cup \Y \cup \Z$, and is thus a $\threeDM$.
\end{proof}
The following corollary follows directly from Lemma~\ref{lem:3DMReduction2}.
\begin{corollary}\label{cor:unbounded}
Let $r=(\vec{\balpha}_m)_{m \in \mathbb{N}^+}$ be an unbounded-value voting rule (Definition~\ref{def:unbounded-value}) that produces a scoring vector with at least $4$ distinct values. Then the $\PW$ problem is NP-complete for $r$.
\end{corollary}

\section{Putting it all together}
\begin{theorem}
Let $r=(\vec{\balpha}_m)_{m \in \mathbb{N}^+}$ be a voting rule that produces a scoring vector with at least $4$ distinct values. Then the $\PW$ problem is NP-complete for $r$.
\end{theorem}
\begin{proof}
If $r$ is an unbounded-value voting rule (Definition~\ref{def:unbounded-value}) then by corollary~\ref{cor:unbounded} the $\PW$ problem is NP-complete for $r$.
Otherwise, $r$ produces a size-$m$ scoring vector $(\vec{\balpha}_m)$ such that every distinct value in $(\vec{\balpha}_m)$ occupies at most $B$ positions for some fixed $B$. In particular, this means that there is a value $a \in (\vec{\balpha}_m)$ where $a>0$, and $a < \alpha_m$ that occupies at most $B$ consecutive positions in $(\vec{\balpha}_m)$. 
But this means that there is an unbounded number of positions in $\vec{\balpha}_m$ outside the range of $a$. By Lemma~\ref{lem:3DMReduction}, the $\PW$ problem is NP-complete for $r$.
\end{proof}

\begin{theorem}
	Let $r=(\vec{\balpha}_m)_{m \in \mathbb{N}^+}$ be a $3$-valued voting rule. Then the $\PW$ problem is NP-complete for $r$ if it produces a size-$m$ scoring vector where the number of positions occupied by values $0$ or $2$ are unbounded.
\end{theorem}
\begin{proof}
This is a direct corollary of Lemma~\ref{lem:3DMReduction2}.
\end{proof}

\bibliographystyle{named}
\bibliography{main}

\input{appendix}

\end{document}

%% file: macros.tex
\usepackage{bbm}

\def\e#1{\emph{#1}}

\def\P{\mathcal{P}}
\def\Q{\mathcal{Q}}
\def\O{\mathcal{O}}
\def\V{\mathcal{V}}

\def\asn{\mathbin{{:}{=}}}

\def\set#1{\mathord{\{#1\}}}

\newcommand{\eat}[1]{}

\newcommand{\algname}[1]{{\sf #1}}
\def\myrulewidth{3.20in}
\def\therule{\rule{\myrulewidth}{0.2pt}}

\newenvironment{insidecode}[3]
{
\begin{tabular}{p{\myrulewidth}}
\multicolumn{1}{c}{\rule{0mm}{3mm}{\bf #3} $\algname{#1}(\mbox{#2})$\vspace{-0.6em}}\\
\therule\vskip-0.8em\therule
\vspace{-1em}
\begin{algorithmic}[1]}
{\end{algorithmic}
\vskip-0.3em\therule
\end{tabular}}

\usepackage{comment}

\def\candidates{\mathcal{C}}

\def\lin{\mathit{lin}}

\def\bX{\boldsymbol{X}}

\def\votes#1#2#3{\mathsf{Votes}_{#3}(#1,#2)}

\def\I{\mathcal{I}}
\def\X{\mathcal{X}}
\def\Y{\mathcal{Y}}
\def\Z{\mathcal{Z}}
\def\S{\mathcal{S}}
\def\s{\mathsf{s}}
\def\V{\mathcal{V}}
\def\Q{\mathcal{Q}}

\def\W{\mathcal{W}}
\def\ol#1{\overline{#1}}

\def\midd#1{\mathbin{|}{#1}}
\def\smax{s^{max}}

\newenvironment{repeatresult}[2]
{\vskip0.5em\par\textsc{#1} #2.\em}
{\vskip1em}

\newenvironment{replemma}[1]{\begin{repeatresult}{Lemma}{#1}}{\end{repeatresult}}

\def\balpha{\boldsymbol{\alpha}}

\def\vote#1{\mathbbm{#1}}

\def\ex3C{\mathrm{\textsc{X3C}}}
\def\threeDM{\mathrm{\textsc{3DM}}}
\def\PW{\mathrm{\textsc{Possible-Winner}}}

%% file: appendix.tex
\section*{Appendix}
\begin{replemma}{\ref{lem:fixcscore}}
\lemmafixcscore
\end{replemma} 
\begin{proof}
	\textbf{The If direction.} Assume that $c$ is a possible winner (co-winner) in $\I^c=(\candidates,\O^c,r)$ where $r$ is a positional voting rule. Then there exists a complete profile $\V$ that is an extension of $\O^c$ where $c$ is the winner (co-winner). But since, by definition, $\lin(\O^c)\subseteq \lin(\O)$, then $\V \in \lin(\O)$, making $c$ a possible winner (co-winner) in $\I$ as well.\\
	\textbf{The Only-If direction.}
	Let $\V=(v_1,\dots,v_n)\in \lin(\O)$ denote the complete profile in which $c$ is the winner (co-winner) of $\I=(\candidates,\O,r)$. 
	If $\V \in \lin(\O^c)$ then we are done. Otherwise, we show how we can transform $\V$ into a complete profile $\V'\in \lin(\O^c)$ where $c$ is the possible winner (co-winner).
	
	Consider any vote $v_i \in \V$ where $v_i \notin \lin(o^c_i)$.
	Since $\O$ is a partitioned profile, then $o_i$ is a partitioned preference. Let $o_i=(A_1 \succ A_2 \succ \dots \succ A_q)$ where 
	the $A_j$s are disjoint partitions.	
	Let $A_l$ ($l \leq q$) be the partition that contains the distinguished candidate $c$ (i.e., $c \in A_l$).
	Since, by our assumption, $o_i$ is a partitioned preference then in any extension $v_i$ of $o_i$, candidate $c$ resides in the range $\set{\sum_{j=1}^{l-1}|A_j|+1,\dots,\sum_{j=1}^{l}|A_j|}$. In particular, the position that maximizes the score of $c$ is $\sum_{j=1}^{l-1}|A_j|+1$. Let $c'$ be the candidate occupying this position in $v_i$. If $c'=c$ then $v_i \in \lin(o^c_i)$ contradicting our assumption. Thus $c' \neq c$.	
	Since all candidates positioned in the range $\set{\sum_{j=1}^{l-1}|A_j|+1,\dots,\sum_{j=1}^{l}|A_j|}$ belong to the partition $A_l$, and in particular the candidate $c'$, then they are also incomparable in $o_i$. Therefore, $c$ and $c'$ are incomparable and can be swapped in $v_i$ while still complying to the partial order $o_i$. 
	
	Let us denote by $v_i'$ the ranking that results from swapping the positions of $c$ and $c'$ in $v_i$. Clearly, we have that (1) $s(v_i,c) \leq s(v'_i,c)$, (2) $s(v_i,c') \geq s(v_i',c')$, and (3) for every $c''\notin \set{c,c'}$ it is the case that $s(v_i,c'')=s(v_i',c'')$.
	
	In the profile $\V'=(v_1,\dots,v_i',\dots,v_n)$ the candidate $c$ is still a winner (co-winner) because its score has increased while the score of any other candidate has not. We repeat this procedure for every vote $v_i$ ($i\in \set{1,\dots,n}$) until we reach a profile that is consistent with $\O^c$ as required.
\end{proof}

For every candidate $a \in \candidates$ we denote by $\votes{a}{1}{\O}$ the set of votes $o\in \O$ in which $\smax(o,a)=1$. Formally: $\votes{a}{1}{\O}=\set{o_i \in \O \mid \smax(o_i,a)=1}$. Clearly $\votes{a}{1}{\O}$ can be computed in polynomial time for every candidate $a\in \candidates$ (even if the profile is not partitioned).
We summarize the winning (co-winning) requirement in the following lemma.
\begin{replemma}{\ref{lem:cWinnerKApproval}}
	\kApprovalLemma
\end{replemma} 
\begin{proof}
	By Lemma~\ref{lem:fixcscore}, candidate $c$ is a PW in a partitioned profile $\O$ if and only if $c$ is a PW in $\I^c=(\candidates, \O^c, r)$.
	Now, $c$ is a PW in $\I^c$ if and only if there exists a complete profile $\V^c \in \lin(\O^c)$ where $s(\V^c,c)>s(\V^c,c')$ for every candidate $c' \neq c$.	
	Note that since $r$ is $k$-approval, we have that:
	\[
	\left|\votes{c'}{1}{\O}\right|=\smax(\O,c')
	\]
	Therefore we can express $s(\V^c,c')$ as follows.
	\begin{align*}
	s(\V^c,c')&=|\votes{c'}{1}{\V^c}|\\
	&\underbrace{=}_{\substack{\votes{c'}{1}{\O}\supseteq \\ \votes{c'}{1}{\V^c}}}|\votes{c'}{1}{\O}\setminus \votes{c'}{0}{\V^c}|\\
	&=|\votes{c'}{1}{\O}|-|\votes{c'}{1}{\O}\cap \votes{c'}{0}{\V^c}|\\
	&=\smax(\O,c')-|\votes{c'}{1}{\O}\cap \votes{c'}{0}{\V^c}|
	\end{align*}
	Candidate $c$ is the winner in $\V^c$ if and only if $s(\V^c,c)>s(\V^c,c')$ for any other candidate $c'$. By definition of $\O^c$ we have that $s(\V^c,c)=\smax(\O,c)$.
	Therefore, $c$ is the winner in $\V^c$ if and only if $\smax(\O,c) > s(\V^c,c')$, that is:
	\[
	\smax(\O,c) > \smax(\O,c')-|\votes{c'}{1}{\O}\cap \votes{c'}{0}{\V^c}|
	\]
	or
	\[
	|\votes{c'}{1}{\O}\cap \votes{c'}{0}{\V^c}|>
	\smax(\O,c')- \smax(\O,c)  
	\]
	as required.
\end{proof}